\documentclass[journal]{IEEEtran}
\usepackage{amsmath}
\usepackage{amssymb}
\usepackage{upref}
\usepackage{amsfonts}
\usepackage{graphicx}
\usepackage{verbatim}
\usepackage{pict2e}

\newcommand{\field}[1]{\mathbb{#1}}

\newcommand{\C}{\field{C}}
\newcommand{\F}{\field{F}}

\newcommand{\T}{\field{T}}

\newcommand{\cA}{{\cal A}}

\newcommand{\cC}{{\cal C}}

\newcommand{\cG}{{\cal G}}

\newcommand{\cP}{{\cal P}}

\newcommand{\cU}{{\cal U}}
\newcommand{\cT}{{\cal T}}

\newcommand{\sP}{\cP}
\newcommand{\sG}{\cG}

\newcommand{\Gr}{\smash{{\sG\kern-1.5pt}_q\kern-0.5pt(n,k)}}
\newcommand{\Grtwo}{\smash{{\sG\kern-1.5pt}_2\kern-0.5pt(n,k)}}
\newcommand{\Gkone}{\smash{{\sG\kern-1.5pt}_q\kern-0.5pt(n,k_1)}}
\newcommand{\Gktwo}{\smash{{\sG\kern-1.5pt}_q\kern-0.5pt(n,k_2)}}
\newcommand{\Ps}{\smash{{\sP\kern-2.0pt}_q\kern-0.5pt(n)}}

\newcommand{\qnom}[2]{\genfrac{[}{]}{0pt}{}{#1}{#2}_q}
\DeclareMathOperator{\codeg}{codeg}

\newcommand{\deff}{\mbox{$\stackrel{\rm def}{=}$}}

\newtheorem{theorem}{Theorem}
\newtheorem{proposition}{Proposition}

\newtheorem{cor}{Corollary}

\begin{document}

\bibliographystyle{IEEEtran}

\title{The Asymptotic Behavior of Grassmannian Codes}
\author{{Simon R. Blackburn} and {Tuvi Etzion,~\IEEEmembership{Fellow,~IEEE}}
\thanks{S. Blackburn is with the Department of Mathematics,
Royal Holloway, University of London,
Egham, Surrey TW20 0EX, United Kingdom.
(email: s.blackburn@rhul.ac.uk).}
\thanks{T. Etzion is with the Department of Computer Science,
Technion --- Israel Institute of Technology, Haifa 32000, Israel.
(email: etzion@cs.technion.ac.il).}
\thanks{This work was supported in part by the Israeli
Science Foundation (ISF), Jerusalem, Israel, under
Grant 230/08.} }

\maketitle
\begin{abstract}
  The iterated Johnson bound is the best known upper bound on a size
  of an error-correcting code in the Grassmannian $\cG_q(n,k)$. The
  iterated Sch\"{o}nheim bound is the best known lower bound on the
  size of a covering code in $\cG_q(n,k)$. We use probabilistic
  methods to prove that both bounds are asymptotically attained for
  fixed $k$ and fixed radius, as $n$ approaches infinity. We also
  determine the asymptotics of the size of the best Grassmannian codes and
  covering codes when $n-k$ and the radius are fixed, as $n$
  approaches infinity.
\end{abstract}


\section{Introduction}
\label{sec:introduction}

\PARstart{L}{et} $\F_q$ be the finite field of order $q$ and let $n$
and $k$ be integers such that $0 \leq k \leq n$. The \emph{Grassmannian}
$\cG_q(n,k)$ is the set of all $k$-dimensional
subspaces of $\F_q^n$. We have that
$$
| \cG_q (n,k)| = \qnom{n}{k} \deff \frac{(q^n-1)(q^{n-1}-1)
\cdots (q^{n-k+1}-1)}{(q^k-1)(q^{k-1}-1) \cdots (q-1)}~,
$$
where $\qnom{n}{k}$ is the $q-$\emph{ary Gaussian binomial
coefficient}. A~natural measure of distance in $\cG_q(n,k)$ is the
\emph{subspace metric}~\cite{AAK01,KoKs08} given by
\[
d_S (U,V) \deff 2k - 2 \dim (U \cap V)
\]
for $U,V \in \cG_q(n,k)$. We say that $\C \subseteq \cG_q(n,k)$ is
an \emph{$(n,M,d,k)_q$ code in the Grassmann space} if $|\C|=M$
and $d_S (U,V) \geq d$ for all distinct $U, V \in \C$. Such a
code~$\C$ is also called a constant dimension code.  The subspaces
in~$\C$ are called \emph{codewords}. (Note that the distance
between any pair of elements of $\cG_q(n,k)$ is even.  Because of
this, some authors define the distance between subspaces $U$ and
$V$ as $\frac{1}{2}d_S(U,V)$.) An important observation is the
following: a code $\C$ in the Grassmann space $\cG_q(n,k)$ has
minimum distance $2\delta+2$ or more if and only if each subspace
in $\cG(n,k-\delta)$ is contained in at most one codeword.  There
is a `dual' notion to a Grassmannian code, known as a $q-$covering
design: we say that $\C \subseteq \cG_q(n,k)$ is a
\emph{$q$-covering design} $\C_q(n,k,r)$ if each element of
$\cG_q(n,r)$ is contained in at least one element of $\C$. If each
element of $\cG_q(n,r)$ is contained in exactly one element of
$\C$, we have a \emph{Steiner structure}, which is both an optimal
Grassmannian code and an optimal $q$-covering
design~\cite{EtVa11,ScEt02}. Codes and designs in the Grassmannian
have been studied extensively in the last five years due to the
work by Koetter and Kschischang~\cite{KoKs08} in random network
coding, who showed that an $(n,M,d,k)_q$ code can correct any $t$
packet insertions and any $s$ packet erasures, as long as $2t+2s
<d$. Our goal in this paper is to examine cases in which we can
determine the asymptotic behavior of codes and designs in the
Grassmannian.

Let $\cA_q(n,d,k)$ denote the maximum number of codewords in an
$(n,M,d,k)_q$ code. The \emph{packing bound} is the best known
asymptotic upper bound for $\cA_q(n,d,k)$. If we write $d=2\delta+2$,
we have
\begin{equation}
\label{eqn:ineq_packing} \cA_q(n,2\delta+2,k) \leq
\frac{\qnom{n}{k-\delta}}{\qnom{k}{k-\delta}}.
\end{equation}
This bound is proved by noting that in an $(n,M,2\delta+2,k)_q$
code, each $(k-\delta)$-dimensional subspace can be contained in
at most one codeword. Bounds on $\cA_q(n,d,k)$ were given in many
papers,
e.g.~\cite{EtSi09,EtSi11,EtVa08,EtVa11,KoKs08,KoKu08,TrRo10,WXS,XiFu09},
In particular, the well-known Johnson bound for constant weight
codes was adapted for constant dimension codes independently
in~\cite{EtVa08,EtVa11,XiFu09} to show that
\[
\cA_q (n,2 \delta+2 ,k)\leq \frac{q^n-1}{q^k-1}
\cA_q(n-1,2\delta+2,k-1).
\]
By iterating this bound, using the observation that
$\cA_q(n,2\delta+2,k)=1$ for all $k \leq \delta$, we obtain the \emph{iterated
Johnson bound}:
\begin{multline*}
\cA_q(n,2\delta+2,k) \\
\leq
\left\lfloor\frac{q^n-1}{q^k-1}
\left\lfloor\frac{q^{n-1}-1}{q^{k-1}-1} \cdots
\left\lfloor\frac{q^{n-k+\delta+1}-1}{q^{\delta+1}-1}
\cdots\right\rfloor\right\rfloor\right\rfloor.
\end{multline*}
It is not difficult to see that
the iterated Johnson bound is always stronger than the
packing bound (indeed, the packing bound may be derived as a
simple corollary of the iterated Johnson bound). However, the main
goal of this paper is to prove that the packing bound (and so the iterated
Johnson bound) is attained asymptotically for fixed $k$ and
$\delta$, $k \geq \delta$, when $n$ tends to infinity. In other
words, we will prove the following theorem, in which the term
$A(n) \sim B(n)$ means that $\lim_{n \rightarrow \infty} A(n)/B(n)
=1$.
\begin{theorem}
\label{thm:packing_asymptotics} Let $q$, $k$ and $\delta$ be fixed
integers, with $0 \leq \delta\leq k$ and such that $q$ is a prime
power. Then
\begin{equation}
\label{eqn_packing}
\cA_q(n,2\delta+2,k)\sim \frac{\qnom{n}{k-\delta}}{\qnom{k}{k-\delta}}
\end{equation}
as $n\rightarrow\infty$.
\end{theorem}
In fact, the proof of our theorem shows a little more than this: see
the proof of the theorem and the comment in the last section of this
paper. Our proof of the lower bound is probabilistic, making use of
some of the theory of quasi-random hypergraphs. There are known
explicit constructions that produce codes whose size is within a
constant factor of the packing bound as
$n\rightarrow\infty$. Currently, the best codes known are the codes of
Etzion and Silberstein~\cite{EtSi09} that are obtained by extending
the codes of Silva, Kschischang, and Koetter~\cite{SKK08} using a
`multi-level construction'. If $q=2$ and $\delta=2$, then the ratio
between the size of the code and the packing bound is 0.6657, 0.6274,
and 0.625 when $k=4$, $k=8$, and $k=30$ respectively, as $n$ tends to
infinity. When $k=3$, the ratio of 0.7101 in~\cite{SKK08} was improved
in~\cite{EtSi11} to 0.7657. The Reed--Solomon-like codes
of~\cite{KoKs08} represented as a lifting of codewords of maximum rank
distance codes~\cite{SKK08} approach the packing bound as
$n\rightarrow\infty$ when one of $\delta$ or $q$ also tends to
infinity~\cite[Lemma~19]{EtSi11}. Theorem~\ref{thm:packing_asymptotics}
shows that there exist codes approaching the packing bound as
$n\rightarrow\infty$ even when $\delta$ and $q$ are fixed; of course,
the challenge is now to construct such codes explicitly.

The paper also proves a similar result for $q$-covering designs.~ Let
$\cC_q(n,k,r)$ denote the minimum
number of $k$-dimensional subspaces in a $q$-covering design
$\C_q(n,k,r)$. Bounds on $\cC_q(n,k,r)$ can be found
in~\cite{Etz11,EtVa11a}. Setting $r=k-\delta$, the \emph{covering bound} states
that

\begin{small}
\begin{equation}
\label{eq:covering}
\cC_q(n,k,r) \geq
\frac{\qnom{n}{k-\delta}}{\qnom{k}{k-\delta}}.
\end{equation}
\end{small}
This bound may be proved by observing that in a $\C_q(n,k,k-\delta)$
covering design each $(k-\delta)$-dimensional subspace must be
contained in at least one codeword. The \emph{Sch\"{o}nheim bound} is an
analogous result to the Johnson bound above:
\[
\cC_q(n,k,r) \geq \frac{q^n-1}{q^k-1} \cC_q(n-1,k-1,r-1).
\]
This bound implies the iterated Sch\"{o}nheim bound~\cite{EtVa11a}:
\begin{equation}
\label{eq:schonheim}
\cC_q(n,k,r) \geq \left\lceil\frac{q^n\!-\!1}{q^k\!-\!1}
\left\lceil\frac{q^{n-1}\!-\!1}{q^{k-1}\!-\!1} \cdots
\left\lceil\frac{q^{n-r+1}\!-\!1}{q^{k-r+1}\!-\!1}\right\rceil
\cdots\right\rceil\right\rceil.
\end{equation}
The iterated Sch\"{o}nheim bound is always at least as strong as the
covering bound. But the following theorem shows that when $k$ and
$\delta$ are fixed with $n\rightarrow\infty$ the covering bound (and
so the iterated  Sch\"{o}nheim bound) is attained asymptotically:
\begin{theorem}
\label{thm:covering_asymptotics} Let $q$, $k$ and $\delta$ be
fixed integers, with ${0\leq \delta\leq k}$ and such that $q$ is a
prime power. Then
\[
\cC_q(n,k,k-\delta)\sim \frac{\qnom{n}{k-\delta}}{\qnom{k}{k-\delta}}
\]
as $n\rightarrow\infty$.
\end{theorem}
The proof of the theorem does not explicitly construct families of
$q$-designs whose ratio with the covering bound approaches $1$.~
The relationship between the best known $q$-covering designs and
the covering bound is more complicated than in the case of
Grassmannian codes, but it is usually the case that better ratios can be
obtained by explicit constructions of $q$-covering designs when
compared to the corresponding problem for Grassmannian codes. For
example, a ratio of 1.05 can be obtained by explicit constructions~\cite{Etz11}
when $q=2$, $k=3$, and $\delta=1$, as $n \rightarrow
\infty$.

The asymptotics of $\cA_q(n,2\delta+2,k)$ when $n-k$ and $\delta$ are
fixed, and of $\cC_q(n,k,r)$ when $n-k$ and $r$ are fixed, are also
determined in this paper. The result for $\cA_q(n,2\delta+2,k)$ is a
simple corollary of Theorem~\ref{thm:packing_asymptotics}, whereas the
result for $\cC_q(n,k,r)$ follows from results in finite geometry.

The rest of the paper is organized as follows. In
Section~\ref{sec:main} we will present the proofs for our main
theorems. In Section~\ref{sec:large_k} we consider the case when $n-k$
is fixed as $n\rightarrow\infty$. Finally, in
Section~\ref{sec:otherAsy} we provide comments on our results, and
state some open questions.

\section{Proofs of the main Theorems}
\label{sec:main}

We begin by observing a simple relationship between the minimum size
of a $q$-covering design and the maximum size of a Grassmannian code.

\begin{proposition}
\label{prop:design_equiv_code}
We have that
\begin{multline*}
\cC_q(n,k,k-\delta)\leq
\cA_q(n,2\delta+2,k)+\\
\qnom{n}{k-\delta}-\qnom{k}{k-\delta}
\cA_q(n,2\delta+2,k)
\end{multline*}
and
\begin{multline*}
\cA_q(n,2\delta+2,k)\geq \cC_q(n,k,k-\delta) +\\
\qnom{n}{k-\delta}-
\qnom{k}{k-\delta}\cC_q(n,k,k-\delta).
\end{multline*}
In particular, Theorems~\ref{thm:packing_asymptotics}
and~\ref{thm:covering_asymptotics} are equivalent.
\end{proposition}
\begin{proof}
Let $\C$ be a Grassmannian code of size $\cA_q(n,2\delta+2,k)$.
There are exactly $\qnom{k}{k-\delta}\cA_q(n,2\delta+2,k)$
subspaces of dimension $k-\delta$ that lie in some element of~$\C$,
since no subspace of dimension $k-\delta$ is contained in
more than one element of $\C$. Thus there are $\Upsilon
\deff\qnom{n}{k-\delta}-\qnom{k}{k-\delta} \cA_q(n,2\delta+2,k)$
uncovered subspaces of dimension $k-\delta$, and we may construct
a $q$-covering design by adding  $\Upsilon$ or fewer
$k$-dimensional subspaces to $\C$. This establishes the first
inequality of the proposition.

To establish the second inequality, let $\C$ be a $q-$covering
design of size $\cC_q(n,k,k-\delta)$. There are
${\qnom{k}{k-\delta}\cC_q(n,k,k-\delta)}$ pairs $(U,V)$ such that
${U\in\cG_q(n,k-\delta)}$, $V\in\C$ and ${U\subseteq V}$. Suppose
we order these pairs in some way. Since every
$(k-\delta)-$dimensional subspace $U$ occurs at least once as the
first element of a pair, there are
$\qnom{k}{k-\delta}\cC_q(n,k,k-\delta)-\qnom{n}{k-\delta}$ pairs
$(U,V)$ where a pair $(U,V')$ for some $V'\in \C$ occurs earlier
in the ordering. Removing the corresponding subspaces $V$ from
$\C$ produces a Grassmannian code of size at least
$\cC_q(n,k,k-\delta) +\qnom{n}{k-\delta}-
\qnom{k}{k-\delta}\cC_q(n,k,k-\delta)$, and so the second
inequality follows.

Suppose Theorem~\ref{thm:packing_asymptotics} holds. Let $q$ be a
fixed prime power, and let $k$ and $\delta$ be fixed integers such
that $0\leq \delta\leq k$.  Then~\eqref{eqn_packing} implies that $\qnom{n}{k-\delta}-\qnom{k}{k-\delta}
\cA_q(n,2\delta+2,k)=o\left(\qnom{n}{k-\delta}\right)$ and so the
first inequality of the proposition implies that
\begin{align*}
\cC_q(n,k,k-\delta)&\leq \cA_q(n,2\delta+2,k)+o\left(\qnom{n}{k-\delta}\right)\\
&\leq \frac{\qnom{n}{k-\delta}}{\qnom{k}{k-\delta}}+o\left(\qnom{n}{k-\delta}\right) ~~~ \text{by}~(\ref{eqn:ineq_packing})\\
&\sim \frac{\qnom{n}{k-\delta}}{\qnom{k}{k-\delta}}.
\end{align*}
Theorem~\ref{thm:covering_asymptotics} now follows from this asymptotic inequality
and the covering bound~\eqref{eq:covering}.

The proof that Theorem~\ref{thm:packing_asymptotics} follows from
Theorem~\ref{thm:covering_asymptotics} is similar to the above, and is omitted.
\end{proof}

We prove Theorem~\ref{thm:packing_asymptotics} by using a result
in quasi-random hypergraphs.~ To state this result,~ we begin by
recalling some terminology from hypergraph theory.~ A hypergraph~$\Gamma$
is \emph{$\ell$-uniform} if all its hyperedges have
cardinality~$\ell$. The \emph{degree} $\deg(u)$ of a vertex
$u\in\Gamma$ is the number of hyperedges containing~$u$; if
$\deg(u)=r$ for all $u\in\Gamma$, we say that $\Gamma$ is
\emph{$r$-regular}. The \emph{codegree} $\codeg(u_1,u_2)$ of a
pair of distinct vertices $u_1,u_2\in\Gamma$ is the number of
hyperedges containing both $u_1$ and $u_2$. A \emph{matching} (or
edge packing) in $\Gamma$  is a set of pairwise disjoint
hyperedges of $\Gamma$. We write $\cU(\Gamma)$ for the minimum
number of vertices left uncovered by a matching in $\Gamma$. Thus
the largest number of hyperedges in a matching of an
$\ell$-uniform hypergraph $\Gamma$ on $v$ vertices is
$(v-\cU(\Gamma))/\ell$. The main theorem we use is due to
Vu~\cite[Theorem~1.2.1]{Vu}:

\begin{theorem}
\label{thm:hypergraph} Let $\ell$ be a fixed integer, where
$\ell\geq 4$. Then there exist constants $\alpha$ and $\beta$ with
the following property. Let $\Gamma$ be an $\ell$-uniform
$r$-regular hypergraph with $v$ vertices. Define $c
=\max\codeg(u_1,u_2)$, where the maximum is taken over all
distinct vertices $u_1,u_2\in \Gamma$. Then
\[
\cU(\Gamma)\leq \alpha v(c/r)^{1/(\ell-1)}(\log r)^\beta.
\]
\end{theorem}

The proof of Theorem~\ref{thm:hypergraph} uses probabilistic methods,
inspired by the techniques of Frankl and
R\"odl~\cite{FrRo85,Rod85}. See~\cite{ABKV03,AlSp08,PiSp89}
for related work.

%

\begin{proof}[Proof of Theorem~\ref{thm:packing_asymptotics}]
If $\delta=0$, then the set of all subspaces in the Grassmannian
is a code that achieves the packing bound; if $\delta=k$ then any
single subspace of dimension $k$ achieves the packing bound. So we
may assume that $0<\delta<k$. Now suppose that $k=2$, so
$\delta=1$.  The theorem follows in this case since it is
known~\cite{EtVa11} that $\cA_q(n,4,2)= \frac{q^n-1}{q^2-1}$ if
$n$ is even; and $\cA_q(n,4,2) \geq \frac{q^n-1}{q^2-1} -
\frac{q^2}{q+1}$ if $n$ is odd. Thus we may suppose that $k\geq
3$.

Define a hypergraph $\Gamma_n$ as follows. We identify the set of vertices of
$\Gamma_n$ with $\cG_q(n,k-\delta)$, and the set of hyperedges of
$\Gamma_n$ with $\cG_q(n,k)$. We define a hyperedge $V$ to
contain a vertex $U$ if and only if $U\subseteq V$ (as
subspaces). We note that $\cA_q(n,2\delta+2,k)$ is exactly the maximum
size of a matching in $\Gamma_n$.

Now $\Gamma_n$ is an $\ell$-uniform hypergraph, where
$\ell=\qnom{k}{k-\delta}$. Note that $\ell\geq 4$, and $\ell$ does
not depend on $n$. Every vertex of $\Gamma_n$ has degree
$r(n)=\qnom{n-(k-\delta)}{\delta}$. Let $U_1$ and $U_2$ be
distinct vertices, so $\dim (U_1+U_2)=k-\delta+i$ for some
positive integer $i$. Then $\codeg(U_1,U_2)$ is the number of
$k-$dimensional subspaces containing $U_1+U_2$, which is at most
the number of $k$-dimensional subspaces containing a
$(k-\delta+1)$-dimensional subspace of $U_1+U_2$. So
\[
\codeg(U_1,U_2) = \qnom{n-(k-\delta+i)}{\delta-i} \leq
\qnom{n-(k-\delta+1)}{\delta-1}.
\]
But
\begin{align*}
\qnom{n-(k-\delta)}{\delta}&=\Theta(q^{n\delta})\text{ and }\\
\qnom{n-(k-\delta+1)}{\delta-1}&=\Theta(q^{n(\delta-1)})
\end{align*}
and so $\max_{u_1,u_2\in\Gamma_n}\codeg(u_1,u_2)=O(q^{-n}r(n))$.
Theorem~\ref{thm:hypergraph} now implies that there exists an
integer $\beta$ such that
\[
\cU(\Gamma_n)=O\left(  \qnom{n}{k-\delta} q^{-n/(\ell-1)}(\log
r(n))^\beta\right).
\]
Thus $\cU(\Gamma_n)=o(\qnom{n}{k-\delta})$, and so the largest
matching in~$\Gamma_n$ contains at least
$\qnom{n}{k-\delta}(1-o(1))/\ell$ edges. The packing bound shows
that the largest matching in $\Gamma_n$ has size at most
$\qnom{n}{k-\delta}/\ell$, and so $\cA(n,2\delta+2,k)\sim
\qnom{n}{k-\delta}/\ell$, as required.
\end{proof}

\begin{proof}[Proof of Theorem~\ref{thm:covering_asymptotics}]
  Theorem~\ref{thm:covering_asymptotics} immediately follows from
  Proposition~\ref{prop:design_equiv_code} and
  Theorem~\ref{thm:packing_asymptotics}.
\end{proof}

\section{The case of large $k$}
\label{sec:large_k}

In the previous section, we assumed that $k$ is fixed (and therefore
is small when compared to $n$). In this section,  we consider the
`dual'  case, where $n-k$ is assumed to be fixed (and so $k$ is
large).

It is proved in~\cite{EtVa11,KoKs08,XiFu09} that
$\cA_q(n,2\delta+2,k)=\cA_q(n,2\delta+2,n-k)$. (This holds because
taking the duals of all subspaces in an $(n,M,d,k)_q$ code in the
Grassmann space produces an $(n,M,d,n-k)_q$-code.) Thus we have the
following corollary of Theorem~\ref{thm:packing_asymptotics}, which
establishes the asymptotics of $\cA_q(n,2\delta+2,k)$ when $n-k$ and
$\delta$ are fixed with $n\rightarrow\infty$.

\begin{cor}
\label{cor:packing_large_k}
Let $q$, $t$ and $\delta$ be fixed integers such that $0\leq
\delta\leq t$, and such that $q$ is a prime power. Then
\begin{equation}
\label{eqn_packing}
\cA_q(n,2\delta+2,n-t)\sim \frac{\qnom{n}{t-\delta}}{\qnom{t}{t-\delta}}
\end{equation}
as $n\rightarrow\infty$.
\end{cor}

Note that when $\delta>t$ we have that
$\cA_q(n,2\delta+2,n-t)=\cA_q(n,2\delta+2,t)=1$, so the restriction on
$\delta$ in Corollary~\ref{cor:packing_large_k} is a natural one.

The same techniques do not establish a similar result for
$q$-covering designs, since $\cC_q(n,k,r)$ and $\cC_q(n,n-k,r)$
are not equal in general. However, by translating some of the results
known in finite geometry into our language, we can determine
$\cC_q(n,k,r)$ when $q$, $r$ and $n-k$ are fixed, as
Theorem~\ref{thm:covering_large_k} below shows.

For the proof of the theorem will need the notion of a
$q-$Tur\'{a}n design. We say that $\C \subseteq \cG_q(n,r)$ is a
\emph{$q$-Tur\'{a}n design} $\T_q(n,k,r)$ if each element of
$\cG_q(n,k)$ contains at least one element of $\C$. Let
$\cT_q(n,k,r)$ denote the minimum number of $r$-dimensional
subspaces in a $q$-covering design $\T_q(n,k,r)$. The notions of
$q$-covering designs and $q$-Tur\'an designs are dual; the
following result was proved in~\cite{EtVa11a}:

\begin{theorem}
\label{thm:coveringTuran} $\cC_q(n,k,r) = \cT_q(n,n-r,n-k)$ for
all ${1 \leq r \leq k \leq n}$.
\end{theorem}

Using normal spreads~\cite{Lun99} (also known as geometric
spreads) Beutelspacher and Ueberberg~\cite{BeUe91} proved the
following theorem using some of the theory of finite projective
geometry.

\begin{theorem}
\label{thm:Tspreads} $\cT_q (vm+\delta,vm-v+1+\delta,m) =
\frac{q^{vm}-1}{q^m-1}$ for all $v \geq 2$ and $m \geq 2$.
\end{theorem}

We remark that Beutelspacher and Ueberberg show much more: that
there is essentially only one optimal construction for a $q$-Tur\'an
design with these parameters.

As a consequence from Theorems~\ref{thm:coveringTuran}
and~\ref{thm:Tspreads} we obtain the following result for
$q$-covering designs.
\begin{cor}
\label{cor:spread}
Let $r$ and $n$ be positive integers such that $r+1$ divides $n$. Then
\[
\cC_q(n,n-n/(r+1),r)=\frac{q^n-1}{q^{n/(r+1)}-1}.
\]
\end{cor}
\begin{proof}
Theorems~\ref{thm:coveringTuran} and~\ref{thm:Tspreads} (in the case
when $\delta=0$) show that
\[
\cC_q(vm,vm-m,v-1)=\frac{q^{vm}-1}{q^m-1}
\]
for any integers $v\geq 2$ and $m\geq 2$. If we set $v=r+1$ and
$m=n/v$, the corollary follows except in the case when $n=2$ and
$r=1$. But the corollary is true in this case also, as a $q$-covering
design with these parameters must consist of all $1$-dimensional subspaces.
\end{proof}

\begin{theorem}
\label{thm:covering_large_k}
Let integers $q$, $t$ and $r$ be fixed, where $q$ is a prime
power. For all sufficiently large integers $n$,
\[
\cC_q(n,n-t,r)=\frac{q^{(r+1)t}-1}{q^t-1}.
\]
\end{theorem}
\begin{proof}
We first note that
\begin{equation}
\label{eq:covering_large_k_eqn}
\cC_q(n+1,n+1-t,r)\leq \cC_q(n,n-t,r).
\end{equation}
This is proved in~\cite{EtVa11a}. To see
why~\eqref{eq:covering_large_k_eqn} holds, fix a $1-$dimensional
subspace $K$ of an ($n+1$)-dimensional vector space $V$. Let $\C$
be a $q$-covering design $\C_q(n,n-t,r)$ contained in the
$n$-dimensional space $V/K$. Then the set of subspaces $U$ such
that $K\subseteq U\subseteq V$ and $U/K\in\C$ is a $q$-covering
design $\C_q(n+1,n+1-t,r)$ containing at most $\cC_q(n,n-t,r)$
subspaces.

The inequality~\eqref{eq:covering_large_k_eqn} implies that for any
fixed $t$ and $r$, we have that $\C_q(n,n-t,r)$ is a non-increasing
sequence of positive integers as $n$ increases. So there exists a
constant $c$ (depending only on $q$, $t$ and $r$) so that
$\C_q(n,n-t,r)=c$ whenever $n$ is sufficiently large. It remains to
show that $c= (q^{(r+1)t}-1)/(q^t-1)$.

Set $n'=t(r+1)$, so $n'-t=n'-n'/(r+1)$. Corollary~\ref{cor:spread}
implies that
\[
c\leq \cC_q(n',n'-t,r)=\frac{q^{n'}-1}{q^{n'/(r+1)}-1}
  =\frac{q^{(r+1)t}-1}{q^t-1}.
\]

Now $c$ is bounded below by the Sch\"onheim bound~\eqref{eq:schonheim}. We give a simpler
form for the Sch\"onheim bound that holds for all sufficiently large $n$ as
follows. When $n$ is sufficiently large we find that
\[
\left\lceil \frac{q^{n-r+1}-1}{q^{k-r+1}-1}\right\rceil=q^t+1=\frac{q^{2t}-1}{q^t-1}.
\]
Moreover, for $i$ such that $0\leq i\leq r-2$,
\[
\left\lceil\frac{q^{n-i}-1}{q^{k-i}-1} \times \frac{q^{(r-i)t}-1}{q^t-1}\right\rceil=\frac{q^{(r-i+1)t}-1}{q^t-1}
\]
provided that $n$ is sufficiently large. These equalities show that
the right hand side of the Sch\"onheim bound~\eqref{eq:schonheim} is equal to
$(q^{(r+1)t}-1)/(q^t-1)$ for all sufficiently large integers $n$. So
$c\geq (q^{(r+1)t}-1)/(q^t-1)$, as required.
\end{proof}

\section{Optimal Codes and Research directions}
\label{sec:otherAsy}

In this section, we comment on our results, we provide a little extra
background, and we propose topics for further study.

We have proved that for a given $q$, if we fix $k$, and $\delta$,
where $\delta < k$, the packing bound for Grassmannian codes is
asymptotically attained when $n$ tends to infinity. We commented in
Section~\ref{sec:introduction} that the same is true when $q$ or
$\delta$ grows. In Section~\ref{sec:large_k}, we determined the
asymptotics of $\cA_q(n,2\delta+2,k)$ when $n-k$ and $\delta$ are
fixed. These results do not address the cases when $q$ and $\delta$
are fixed, but $k$ and $n-k$ both grow (for example when $k=\lfloor
\alpha n\rfloor$ for some fixed real number $\alpha\in (0,1)$). Can
similar results be obtained a wide range of these cases? When $k$
grows rather slowly when compared to $n$, it should be possible to use
a result of Alon et al~\cite{ABKV03} to show that
$\cA_q(n,2\delta+2,k)$ still approaches the packing bound.

The proof of Theorem~\ref{thm:packing_asymptotics} does not just give
the leading term of $\cA_q(n,2\delta+2,k)$: the order of the error
term is also given. However, we do not see any reason why this error term is
tight.

Similar questions can be asked about the relationship between the
covering bound and $\cC_q(n,k,r)$.  It seems that small $q$-covering
designs are easier to construct than large Grassmannian codes;
certainly there are more construction methods currently
known~\cite{Etz11,EtVa11a}.

As well as trivial cases, there are a few sets of parameters for
which the exact (or almost the exact) values of $\cA_q (n,d,k)$
and $\cC_q(n,k,r)$ are known.
Section~\ref{sec:large_k} discusses a family of optimal
$q$-covering designs. A family of optimal Grassmannian codes is known
when $d=2k$. \emph{Spreads} (from projective geometry)
give rise to optimal codes as well as $q$-covering designs when $k$
divides $n$. Known \emph{partial spreads} of maximum size give rise to
optimal codes in other
cases~\cite{DeMe07,ESS,EJSSS,GaSz03}.


For small parameters, the best known codes are very often cyclic
  codes, which are defined as follows. Let $\alpha$ be a primitive element of
GF($q^n$). We say that~a~code $\C \subseteq \cG_q(n,k)$ is
\emph{cyclic} if it has the following property: whenever $\{{\bf
  0},\alpha^{i_1},\alpha^{i_2},\ldots,\alpha^{i_m}\}$ is a codeword of
$\C$, so is its cyclic~shift $\{{\bf 0},
\alpha^{i_1+1},\alpha^{i_2+1},\ldots,\alpha^{i_m+1} \}$. In other
words, if we map each subspace $V \,{\in}\, \C$ into the
corresponding binary characteristic vector $x_V =
(x_0,x_1,\ldots,x_{q^n-2})$ given by
$$
x_i = 1 ~~\text{if $\alpha^i {\in}\kern1pt V$}
\hspace{4ex}\text{and}\hspace{4ex} x_i = 0 ~~\text{if $\alpha^i
{\not\in}\, V$}
$$
then the set of all such characteristic vectors is closed under cyclic
shifts. It would be very interesting to find out whether cyclic codes
approach the packing bound and the covering bound asymptotically.
Again, in this case we would like to see proofs similar to the ones of
Theorems~\ref{thm:packing_asymptotics}
and~\ref{thm:covering_asymptotics}. Of course, explicit families of
asymptotically good cyclic codes would be even more worthwhile.

\begin{center}
{\bf Acknowledgement}
\end{center}
The authors would like to thank Simeon Ball for introducing to
them the concept of normal spreads.


\begin{thebibliography}{10}
\providecommand{\url}[1]{#1}
\csname url@rmstyle\endcsname
\providecommand{\newblock}{\relax}
\providecommand{\bibinfo}[2]{#2}
\providecommand\BIBentrySTDinterwordspacing{\spaceskip=0pt\relax}
\providecommand\BIBentryALTinterwordstretchfactor{4}
\providecommand\BIBentryALTinterwordspacing{\spaceskip=\fontdimen2\font plus
\BIBentryALTinterwordstretchfactor\fontdimen3\font minus
  \fontdimen4\font\relax}
\providecommand\BIBforeignlanguage[2]{{%
\expandafter\ifx\csname l@#1\endcsname\relax
\typeout{** WARNING: IEEEtran.bst: No hyphenation pattern has been}%
\typeout{** loaded for the language `#1'. Using the pattern for}%
\typeout{** the default language instead.}%
\else
\language=\csname l@#1\endcsname
\fi
#2}}


\bibitem{AAK01}
R.\ Ahlswede, H.\,K.\ Aydinian, and L.\,H.\ Khachatrian,
``On perfect codes and related concepts,'' \emph{Designs, Codes, Crypt.}, vol.
22, pp. 221--237, 2001.

\bibitem{ABKV03} N. Alon, B. Bollobas, J.H. Kim and V.H. Vu,
  ``Economical covers with geometric applications,''
  \emph{Proc. London Math. Soc.}, vol. 86, pp. 273--301, 2003.

\bibitem{AlSp08}
    N. Alon and J. H. Spencer,
    {\em The Probabilistic Method},
    3rd edition, John Wiley \& Sons, Hoboken, 2008.

\bibitem{DeMe07}
    J. de Beule and K. Metsch,
    ``The maximum size of a partial spread in $H(5,q^2)$ is $q^3+1$,''
    \emph{J. Comb. Theory, Ser. A,} vol. 114, pp. 761--768, 2007.

\bibitem{BeUe91}
     A. Beutelspacher and J. Ueberberg,
     ``A characteristic property of geometric $t$-spreads
       in finite projective spaces,''
     \emph{Europ. J. Comb.,} vol. 12, pp. 277--281, 1991.

\bibitem{ESS}
     J. Eisfeld, L. Storme, and P. Sziklai,
     ``On the spectrum of the sizes of maximal partial line
      spreads in $PG(2n,q)$, $n \geq 3$,''
     \emph{Designs, Codes, and Cryptography,} vol. 36, pp. 101--110, 2005.

\bibitem{EJSSS}
     S. El-Zanati, H. Jordon, G. Seelinger, P. Sissokho, and L. Spence,
     ``The maximum size of a partial 3-spread in a finite vector space over $GF(2)$,''
     \emph{Designs, Codes, and Cryptography,} vol. 54, pp. 101--107, 2010.

\bibitem{Etz11}
T. Etzion, ``Covering subspaces by subspaces'', in preparation.

\bibitem{EtSi09}
T. Etzion and N. Silberstein, ''Error-correcting codes in
projective space via rank-metric codes and Ferrers diagrams'',
\emph{IEEE Trans.\ Inform. Theory}, vol. 55, no.7, pp. 2909--2919,
July 2009.

\bibitem{EtSi11}
T. Etzion and N. Silberstein, ``Codes and Designs Related to Lifted MRD Codes,''
\emph{arxiv.org/abs/1102.2593}.


\bibitem{EtVa08}
T. Etzion and A. Vardy, ``Error-correcting codes in projective
space'', in proceedings of \emph{International Symposium on
Information Theory}, pp. 871--875, July 2008.

\bibitem{EtVa11}
T. Etzion and A. Vardy, ``Error-correcting codes in projective
space'', \emph{IEEE Trans.\ Inform. Theory}, vol.\,57, no.\,2,
pp.\,1165--1173, February 2011.

\bibitem{EtVa11a}
T. Etzion and A. Vardy, ``On $q$-Analogs for Steiner Systems and
Covering Designs'', \emph{Advances in Mathematics of
Communications}, vol. 5, no. 2, pp. 161--176, 2011.

\bibitem{GaSz03}
     A. G\'{a}cs and T. Sz\H{o}nyi,
     ``On maximal partial spreads in $PG(n,q)$,''
     \emph{Designs Codes Crypt.,} vol. 29, pp. 123--129, 2003.

\bibitem{FrRo85}
P. Frankl and V. R\"{o}dl, ``Near perfect coverings in graphs and hypergraphs'',
\emph{European J. Combin.}, vol. 6, pp. 317--326,~1985.

\bibitem{KoKs08}
R.\ Koetter and F.\,R.\ Kschischang, ``Coding for errors and
erasures~in~random network coding,'' \emph{IEEE Trans.\ Inform.
Theory}, vol.\,54, no.\,8, pp.\,3579--3591, August 2008.

\bibitem{KoKu08}
A.\,Kohnert and S.\,Kurz, ``Construction of large
constant-dimension~codes with a prescribed minimum distance,''
\emph{Lecture Notes in Computer Science}, vol.\,5393, pp.\,31--42,
December 2008.

\bibitem{Lun99}
G. Lunardon, ``Normal spreads'', \emph{Geometriae Dedicata}, vol.
75, pp. 245--261, 1999.

\bibitem{PiSp89}
N. Pippenger and J. Spencer, ``Asymptotic behavior of the chromatic index for hypergraphs'',
\emph{J. Comb. Theory, Ser. A}, vol. 51, pp. 24--42, 1989.

\bibitem{Rod85}
V. R\"odl, ```On a packing and covering
problem'', \emph{Europ. J. Comb.},
vol. 6, pp. 69--78, 1985.

\bibitem{ScEt02}
M.\ Schwartz and T.\ Etzion, ``Codes and anticodes in the
Grassman graph'', \emph{J.\ Comb.\ Theory, Ser.\,A}, vol. 97, pp. 27--42, 2002.

\bibitem{SKK08}
D. Silva, F.\,R.\ Kschischang, and R.\ Koetter, ``A rank-metric
approach to error control in random network coding,'' \emph{IEEE
Trans. Inform. Theory},  vol.~54, pp. 3951--3967, September
2008.

\bibitem{TrRo10}
A.-L. Trautmann  and J. Rosenthal, ``New improvements on the
echelon-Ferrers construction'', in proc. of \emph{Int. Symp. on
Math. Theory of Networks and Systems}, pp. 405--408, July 2010.

\bibitem{Vu}
Van H. Vu, ``New bounds on nearly perfect matchings in hypergraphs:
Higher codegrees do help'', \emph{Random Structures and Algorithms},
vol.\,17, pp.\,29--63, 2000.

\bibitem{WXS}
H. Wang, C. Xing C and R. Safavi-Naini, ``Lee metric codes over integer residue rings'',
\emph{IEEE Trans. on Inform. Theory}, vol.\,IT-49, pp.\,866--872,
2003.

\bibitem{XiFu09}
S.-T. Xia and F.-W. Fu, ``Johnson type bounds on
constant dimension codes'', \emph{Designs, Codes
Crypt.}, vol. 50, pp. 163--172,~2009.





\end{thebibliography}

\end{document}